\documentclass[11pt]{amsart}
\usepackage{mathrsfs}
\usepackage{amssymb} 
\usepackage{url}
\usepackage{hyperref}
\usepackage{tikz}
\usetikzlibrary{matrix,arrows}

\theoremstyle{definition}
\newtheorem{theorem}{Theorem}[section]
\newtheorem{definition}[theorem]{Definition}
\newtheorem{proposition}[theorem]{Proposition}
\newtheorem{lemma}[theorem]{Lemma}
\newtheorem{corollary}[theorem]{Corollary}
\newtheorem{remark}[theorem]{Remark}
\newtheorem{example}[theorem]{Example}
\newtheorem{notation}[theorem]{Notation}

\newcommand{\reals}{\mathbb{R}}

\newcommand{\sref}[1]{\S\ref{#1}}

\newcommand{\rank}{\textnormal{rank }}
\newcommand{\diag}{\textnormal{diag}}

\newcommand{\tn}{\textnormal}
\newcommand{\ds}{\displaystyle}

\newcommand{\ie}{\textit{i.e.} }
\newcommand{\cf}{\textit{cf.} }
\newcommand{\eg}{\textit{e.g.} }
\newcommand{\citep}[2]{\cite{#1}, p. #2}

\newcommand{\cfeg}[2]{(\cf \eg \citep{#1}{#2})}
\newcommand{\seep}[2]{(see \citep{#1}{#2})}
\newcommand{\seepeg}[2]{(see \eg \citep{#1}{#2})}

\newcommand{\Aut}{\tn{Aut}}

\newcommand{\Sym}{\tn{Sym}}

\newcommand{\mc}[1]{\mathcal{#1}}

\newcommand{\vecgen}[1]{\tn{span}(#1)}
\newcommand{\IM}{\tn{im }}
\newcommand{\tensors}[3]{\mc T{}^{#1}_{#2}(#3)}

\newcommand{\subannih}[1]{{#1}^\diamond{}}

\newcommand{\idxannih}[2]{#1{}^{#2}{}}
\newcommand{\idxcoannih}[2]{#1{}_{#2}{}}

\newcommand{\radix}[1]{\idxcoannih{#1}{\circ}}
\newcommand{\annih}[1]{\idxannih{#1}{\bullet}}
\newcommand{\coannih}[1]{\idxcoannih{#1}{\bullet}}
\newcommand{\coradix}[1]{\idxannih{#1}{\circ}}
\newcommand{\coannihinc}[1]{\idxcoannih{#1}{\hat{\bullet}}{}}
\newcommand{\coradixinc}[1]{\idxannih{#1}{\hat{\circ}}{}}

\newcommand{\annihg}{\coannih{g}}
\newcommand{\coannihg}{\annih{g}}
\newcommand{\metric}[1]{\left\langle#1\right\rangle}
\newcommand{\annihprod}[1]{\coannih{\langle\!\langle#1\rangle\!\rangle}}
\newcommand{\coannihprod}[1]{\annih{\langle\!\langle#1\rangle\!\rangle}}

\newcommand{\factflat}{\flat}
\newcommand{\factsharp}{\sharp}

\newcommand{\idxlow}{\flat}

\newcommand{\idxraise}{\sharp}
\newcommand{\cocontr}{{{}_\bullet}}

\newcommand{\dsfrac}[2]{\ds{\frac{#1}{#2}}}

\def\hyph{-\penalty0\hskip0pt\relax}
\newcommand{\semiriem}{semi{\hyph}Riemannian}

\newcommand{\radannih}{radical{\hyph}annihilator}

\begin{document}

\bibliographystyle{plain}
\title{Tensor Operations on Degenerate Inner Product Spaces}

\author{Ovidiu Cristinel \ Stoica}

\begin{abstract}
Well-known operations defined on a non-degenerate inner pro\-duct vector space are extended to the case of a degenerate inner product. The main obstructions to the extension of these operations to the degenerate case are (1) the index lowering operation is not invertible, and (2) we cannot associate to the inner product in a canonical way a reciprocal inner product on the dual of the vector space. This article shows how these obstructions can be avoided naturally, allowing a canonical definition of covariant contraction for some important special cases.

The primary motivation of this article is to lay down the algebraic foundation for the construction of invariants in Singular Semi-Riemannian Geometry, especially those related to the curvature. It turns out that the operations discussed here are enough for this purpose \cite{Sto11a,Sto11b,Sto11d}. Such invariants can be applied to the study of singularities in the theory of General Relativity \cite{Sto11e,Sto11f,Sto11g,Sto11c,Sto11h}.
\end{abstract}


\maketitle
\setcounter{tocdepth}{1}
\tableofcontents

\section{Introduction and motivation}

On a non-degenerate inner product space $(V,g)$ and on the spaces associated to it we can define various structures and operations. The inner product $g$ induces on the dual $V^*$ a unique inner product \cfeg{Rom08}{59}. It also defines a canonical isomorphism $\idxlow:V\to V^*$ between $V$ and its dual, as well as its inverse $\idxraise$ (see \eg \citep{Gibb06}{15}; \citep{GHLF04}{72}). These two isomorphisms can be extended canonically to the tensor products involving $V$ and $V^*$, and can be used to switch between vectors and covectors -- in other words to lower and raise indices\footnote{In general, we will use the terms ``index'' or ``indices'' in connection to the abstract index notation (see \eg \cite{PeR87}, Chapter 2), which is invariant and independent of the basis. When we will use them as labels for the components of the vectors and tensors in a basis, we will specify this.} \seepeg{ONe83}{81--83}. The inner product can also be used to define tensor contractions between slots of the same type, that is, which are either both covariant or both contravariant \seepeg{ONe83}{83}. These operations are necessary when dealing with tensors, in areas like differential geometry, mechanics, general relativity.

If the inner product $g$ is degenerate, we can still lower indices using $\idxlow$, which is no longer an isomorphism. But we cannot raise indices and define contractions between two covariant indices, at least not in the usual way. 

In this article we present a natural generalization of the operations mentioned above for the non-degenerate inner product spaces, to the degenerate ones. It is easy to see that there is a canonic inner product $\annihg$ defined on the subspace $\flat(V)\subseteq V^*$. We use this inner product to define contraction between covariant slots which cancel on the degenerate space $\radix{V}:=V^\perp$ of $V$. This contraction is uniquely defined, being thus an invariant operation. Also we can define an index raising operator $\idxraise:\flat(V)\to V$, which is not unique, but all the possible solutions are easy to identify. A canonical index raising operator can be defined, which is not valued in $V$, but in $\flat(V)^*$.

The motivation of this research is two-fold. On the one hand, it is the study of singular {\semiriem} manifolds (\ie differentiable manifolds having on the tangent bundle a degenerate metric tensor \cite{Kup87a,Kup87b,Kup87c,Kup96}). More precisely, the purpose is to construct and study invariants similar to those known from the non-singular {\semiriem} geometry but whose construction is obstructed by the degeneracy of the inner product, and to study their properties. This part was developed so far in \cite{Sto11a,Sto11b,Sto11d}. On the other hand, from the viewpoint of applications, the long-term goal is the study of some special types of singularities in General Relativity. We did this for the black hole singularities \cite{Sto11e,Sto11f,Sto11g,Sto11c} and for the Big Bang singularity of the Friedmann-Lema\^itre-Robertson-Walker spacetime \cite{Sto11h}. While we already applied the main results presented here to obtain the mentioned results, the necessity of a more detailed development of the algebraic part led us to this exposition.

\section{Generalities about degenerate inner products}

In the entire article it will be considered that the vector spaces are finite dimensional and real. We review some elementary notions on vector spaces with degenerate inner product, which are known from the literature (\cf \eg \cite{Rom08} Chapter 11, \citep{Kup87b}{260--261} and \citep{Pam03}{262--265}). 

\begin{definition}
\label{def_inner_prod}
An \textit{inner product} on a vector space $V$ is a symmetric bilinear form $g\in V^*\odot V^*:=\Sym(V^*\otimes V^*$). The pair $(V,g)$ is named \textit{inner product space}. We use alternatively the notation $\metric{u,v}:=g(u,v)$, for $u,v\in V$. The inner product $g$ is \textit{degenerate} if there is a vector $v\in V$, $v\neq 0$, so that $\metric{u,v}=0$ for all $u\in V$, otherwise $g$ is \textit{non-degenerate}. There is always a basis, named \textit{orthonormal basis}, in which $g$ takes a diagonal form:
\begin{equation}
	g =  \left[
\begin{array}{cccc}
	& O_r & & \\
	&  & -I_s &  \\
	&  &  & +I_t \\
\end{array}
\right].
\end{equation}
where $O_r$ is the zero operator on $\reals^r$, and $I_q$, $q\in\{s,t\}$ is the identity operator in $\reals^q$.
The \textit{signature} of $g$ is defined as the triple $(r,s,t)$.
\end{definition}

In the following $(V,g)$ is an inner product space.

\begin{definition}
Two sets of vectors $S\subseteq V$ and $S'\subseteq V$ are said to be \textit{orthogonal}, $S\perp S'$, if $\metric{s,s'}=0$ for any $s\in S$, $s'\in S'$. If one or both sets reduce to one element, we may simply write $v\perp S$ or $v\perp v'$, for $S\subseteq V$, $v,v'\in V$.
\end{definition}

\begin{definition}
Let $U\subseteq V$ be a vector subspace.
If there is a vector $u\in U,u\neq 0$, so that $u\perp U$, $u$ is said to be a \textit{degenerate} vector, and $U$ a \textit{degenerate} vector subspace. If $u\neq 0$ and $u\perp V$, then $u$ is said to be \textit{totally degenerate}. If $u\perp u$, $u$ is said to be \textit{null}.
$U\subseteq V$ is degenerate if and only if $g|_U$ is degenerate (even though $g$ may be non-degenerate on $V$).
\end{definition}

\begin{notation}
For any set of vectors $S\subseteq V$, let's denote by $S^{\perp_V}:=\{v\in V|v\perp S\}$ its \textit{orthogonal complement}. When $V$ is understood, we will simply use $S^\perp$ instead of $S^{\perp_V}$.
\end{notation}

\begin{proposition}\seep{Pam03}{264, Proposition 1.5}
\label{thm_subset_perp}
Let $S\subset V$ be a set of vectors from $V$. Then
\begin{enumerate}
	\item \label{thm_subset_perp:vect_space}
	$S^\perp\subseteq V$ is a vector subspace of $V$.
	\item \label{thm_subset_perp:perp_perp}
	$S\subseteq (S^\perp)^\perp$.
	\item \label{thm_subset_perp:vect_space_perp}
	$S^\perp = (\vecgen S)^\perp$, where $\vecgen S$ is the vector space generated by the set $S\subseteq V$.
\end{enumerate}
\end{proposition}
\begin{proof}
\eqref{thm_subset_perp:vect_space}
Let $s\in S$. If $v_1,v_2\in S^\perp$ and $a_1,a_2\in \reals$, then
\begin{equation*}
	\metric{a_1 v_a + a_2 v_2,s}=a_1 \metric{v_a,s} + a_2 \metric{v_2,s}=0.
\end{equation*}
It follows that $a_1 v_a + a_2 v_2 \in S^\perp$, and $S^\perp\subseteq V$ is a vector subspace.

\eqref{thm_subset_perp:perp_perp}
Let $s\in S$. Then, for any $v\in S^\perp$, $\metric{s,v}=0$, so $s\in (S^\perp)^\perp$.

\eqref{thm_subset_perp:vect_space_perp}
We have $S\subseteq \vecgen S$. Let $v\perp S$ and $s=\sum_{i=0}^k a_i s_i$, where $a_i\in\reals$ and $s_i\in S$. Then, $\metric{v,\sum_{i=0}^k a_i s_i}=\sum_{i=0}^k a_i \metric{v,s_i}=0$, so $v\perp s$. Therefore, $S^\perp\subseteq(\vecgen S)^\perp$. Since $s\in S$ implies that $s\in\vecgen S$, any vector $v\perp \vecgen S$ also satisfies $v\perp S$, so we have $S^\perp = (\vecgen S)^\perp$.
\end{proof}

\begin{proposition}
\label{thm_subsets_perp}
Let $S, S'\subseteq V$ be two sets of vectors. Then:
\begin{enumerate}
	\item \label{thm_subsets_perp:inclusion}
	If $S\subseteq S'$, then $S'^\perp\subseteq S^\perp$.
	\item \label{thm_subsets_perp:cup}
	$(S\cup S')^\perp = S^\perp \cap S'^\perp$.
		\item \label{thm_subsets_perp:cap}
$S^\perp + S'^\perp\subseteq(S\cap S')^\perp$.
\end{enumerate}
\end{proposition}
\begin{proof}
	\eqref{thm_subsets_perp:inclusion}
If $v\perp S'$, then $v\perp S$ too, and $S'^\perp\subseteq S^\perp$.

	\eqref{thm_subsets_perp:cup}
$v\in(S\cup S')^\perp$ $\Leftrightarrow$ ($v\perp S$ and $v\perp S'$) $\Leftrightarrow$ ($v\in S^\perp$ and $v\in S'^\perp$) $\Leftrightarrow$ $v\in S^\perp \cap S'^\perp$.

	\eqref{thm_subsets_perp:cap}
We have from Proposition \ref{thm_subsets_perp} \eqref{thm_subsets_perp:inclusion} that $S^\perp\subseteq (S\cap S')^\perp$ and $S'^\perp\subseteq (S\cap S')^\perp$, therefore $S^\perp + S'^\perp\subseteq (S\cap S')^\perp$.
\end{proof}

\begin{definition}(\cf \eg \citep{Bej95}{1}, \citep{Kup96}{3} and \citep{ONe83}{53})
The totally degenerate space $\radix{V}:=V^\perp$ is named the \textit{radical} of $V$. An inner product $g$ on a vector space $V$ is non-degenerate if and only if $\radix{V}=\{0\}$.
\end{definition}

\begin{proposition}(\cf \citep{Pam03}{263}, Proposition 1.4, \citep{ONe83}{49, Lemma 22})
\label{thm_radical}
The radical of $(V,g)$ has, for any $U\subseteq V$, the following properties:
\begin{enumerate}
	\item \label{thm_radical:in_all}
	$\radix{V}\subseteq U^\perp$.
	\item \label{thm_radical:perp_perp}
$(U^\perp)^\perp = U + \radix{V}$.
	\item \label{thm_radical:dim_subspace}
$\dim V = \dim U + \dim U^\perp + \dim(\radix{V} \cap U)$.
\end{enumerate}
\end{proposition}
\begin{proof}
See \cite{Pam03}, Proposition 1.4.
\end{proof}

\section{The radical-annihilator inner product space}
\label{s_dual_inner_prod}

Because $\flat:V\to V^*$ is not an isomorphism, we can no longer define a dual for $g$ on $V^*$ in the usual sense. We will see that we can still define canonically an inner product $\annihg\in\flat(V)^*\odot\flat(V)^*$, and use it to define contraction and index raising in a weaker sense than in the non-degenerate case. This construction is rather elementary in Linear Algebra, but it will prove to be efficient.

\subsection{Subspaces, annihilators, quotient spaces and duality}
\label{s_dual_subspace}

This sections recalls a lemma establishing a general property of vector subspaces (see \eg \citep{Rom08}{102}, \citep{hal99}{26--27, 33-34}).

Recall that if $W\subseteq V$ is a vector subspace of a vector space $V$, then the following sequence is exact
\begin{center}
\begin{tikzpicture}
\matrix (m) [matrix of math nodes, row sep=4em,
column sep=4em, text height=1.5ex, text depth=0.25ex]
{ 0  & W & V & \dsfrac V W & 0 \\};
\path[right hook->]
(m-1-1) edge (m-1-2)
(m-1-2) edge node[auto] {$i$}(m-1-3);
\path[->>]
(m-1-3) edge node[auto] {$\pi$}(m-1-4)
(m-1-4) edge(m-1-5);
\end{tikzpicture}
\end{center}
where $i:W\to V$ is the canonical injection, and $\pi:V\to\frac V W$, $\pi(v)=[v]$.

\begin{definition}\cfeg{Rom08}{102}
\label{def_dual_subspace}
Let $V$ be a vector space and $W\subseteq V$ a vector subspace. The set $\subannih{W}:=\{\omega\in V^*|\omega(v)=0(\forall v\in W)\}$ is a vector subspace named the \textit{annihilator} of $W$.
\end{definition}

\begin{lemma}
\label{thm_dual_subspace}
The annihilator $\subannih{W}$ of a vector subspace $W\subseteq V$ of a vector space $V$ can be identified with $(\frac V W)^*$.
\end{lemma}
\begin{proof}
For any $v,v'\in V$ so that $v-v'\in W$, and for any linear form $\omega\in \subannih{W}$, $\omega(v)=\omega(v')$, hence we can define $\widetilde\omega\in(\frac V W)^*$ by $\widetilde\omega([v]):=\omega(v)$. Reciprocally, any form $\widetilde\omega\in(\frac V W)^*$ can be considered to act on $V$ and annihilate $W$.
\end{proof}

\begin{remark}\rm
\label{thm_img_ker_sub}
$\dim W + \dim\subannih{W}=n$, because $\subannih{W}\cong\IM\pi$ and $W=\ker{\pi}$.
\end{remark}

\begin{remark}\rm
\label{rem_dual_subspace}
The Lemma \ref{thm_dual_subspace} can be applied to $\subannih{W}\subseteq V^*$ to identify $(\frac {V^*}{\subannih{W}})^*$ with $W$. We have another exact sequence:
\begin{center}
\begin{tikzpicture}
\matrix (m) [matrix of math nodes, row sep=4em,
column sep=4em, text height=1.5ex, text depth=0.25ex]
{ 0  & \subannih{W} & V^* & \dsfrac {V^*} {\subannih{W}} & 0 \\};
\path[right hook->]
(m-1-1) edge (m-1-2)
(m-1-2) edge node[auto] {$i'$}(m-1-3);
\path[->>]
(m-1-3) edge node[auto] {$\pi'$}(m-1-4)
(m-1-4) edge(m-1-5);
\end{tikzpicture}
\end{center}
The two sequences are connected by the operation of taking the dual of a vector space:
\begin{center}
\begin{tikzpicture}
\matrix (m) [matrix of math nodes, row sep=4em,
column sep=4em, text height=1.5ex, text depth=0.25ex]
{ 0  & W & V & \dsfrac V W & 0 \\
0	& \dsfrac {V^*} {\subannih{W}} & V^*	& \subannih{W} & 0	\\ };
\path[right hook->]
(m-1-1) edge (m-1-2)
(m-1-2) edge node[auto] {$i$}(m-1-3);
\path[->>]
(m-1-3) edge node[auto] {$\pi$}(m-1-4)
(m-1-4) edge(m-1-5);
\path[->>]
(m-2-2) edge (m-2-1)
(m-2-3) edge node[auto,swap] {$\pi'$}(m-2-2);
\path[left hook->]
(m-2-4) edge node[auto,swap] {$i'$}(m-2-3)
(m-2-5) edge(m-2-4);
\path[dashed]
(m-1-2) edge node[auto] {$\ast$} (m-2-2)
(m-1-3) edge node[auto] {$\ast$} (m-2-3)
(m-1-4) edge node[auto] {$\ast$} (m-2-4);
\end{tikzpicture}
\end{center}
\end{remark}

The following simple lemma will be useful in \sref{s_radical_bases}.
\begin{lemma}
\label{thm_annihilator_dual_bases}
Let $(e_a)_{a=1}^n$ be a basis of a vector space $V$, $\dim V=n$, and let $(\omega^b)_{b=1}^n$ be its dual basis \cfeg{Rom08}{96}. Then, $(e_a)_{a=1}^n$ extends a basis $(e_a)_{a=1}^k$ of a vector subspace $W\subseteq V$, $\dim W=k$ if and only if  $(\omega^b)_{b=k+1}^{n}$ is a basis of $\subannih{W}$, the annihilator of $W$.
\end{lemma}
\begin{proof}
Since the dual basis is the unique basis of $V^*$ so that $\omega^b(e_a)=\delta^b_a$, it follows that $\omega^b(e_a)=0$ for $b>k$ and $a\leq k$. This implies that $\subannih{\vecgen{(e_a)_{a=1}^k}}=\vecgen{(\omega^b)_{b=k+1}^n}$.
\end{proof}

\subsection{The radical-annihilator space}
\label{s_rad_annih_space}

This section applies the well-known notions recalled in \sref{s_dual_subspace}, and other elementary properties of linear algebra (\cf \eg \cite{Rom08}, \cite{hal99}). Its purpose is to extend fundamental notions related to the non-degenerate inner product $g$ on a vector space $V$ induced on the dual space $V^*$ \cfeg{Rom08}{59}, to the case when $g$ is allowed to be degenerate. Let $(V,g)$ be an inner product space over $\reals$.

\begin{definition}[see \eg \citep{Gibb06}{15}; \citep{GHLF04}{72}]
\label{def_inner_morphism}
The inner product $g$ defines a vector space morphism, named the \textit{index lowering morphism} $\flat:V\to V^*$, by associating to any $u\in V$ a linear form $\flat(u):V\to \reals$ defined by $\flat(u)v:=\metric{u,v}$. Alternatively, it is used the notation $u^\flat$ for $\flat(u)$. For reasons which will become apparent, we will also use the notation $\annih u:=u^\flat$.
\end{definition}

\begin{remark}\rm
\label{thm_radix_ker}
It is easy to see that $\radix{V}=\ker\flat$, so $\flat$ is an isomorphism if and only if $g$ is non-degenerate.
\end{remark}

\begin{definition}
\label{def_radical_annihilator}
The \textit{{\radannih}} vector space $\annih{V}:=\IM\flat\subseteq V^*$ is the space of $1$-forms $\omega$ which can be expressed as $\omega=\annih u$ for some $u$, and they act on $V$ by $\omega(v)=\metric{u,v}$.
\end{definition}

Obviously, in the case when $g$ is non-degenerate, we have the identification $\annih{V}=V^*$.

\begin{remark}\rm
\label{thm_img_ker_radix}
In other words, $\annih{V}=\subannih{(\radix{V})}$. It follows then from Remark \ref{thm_img_ker_sub} that $\dim\annih{V}+\dim\radix{V}=n$.
\end{remark}

\begin{remark}\rm
Any $u'\in V$ satisfying $\annih{u'}=\omega$ differs from $u$ by $u'-u\in\radix{V}$. Such $1$-forms $\omega\in\annih{V}$ satisfy $\omega|_{\radix{V}} = 0$.
\end{remark}

\begin{definition}
\label{def_co_inner_product}
On the vector space $\annih{V}$ we can define a unique non-degene\-rate inner product $\annihg$ by $\annihg(\omega,\tau):=\metric{u,v}$, where $\annih u=\omega$ and $\annih v=\tau$. We alternatively use the notation $\annihprod{\omega,\tau}=\annihg(\omega,\tau)$.
\end{definition}

\begin{proposition}
$\annihg$ from above is well-defined.
\end{proposition}
\begin{proof}
If $u',v'\in V$ are other vectors satisfying $\annih{u'}=\omega$ and $\annih{v'}=\tau$, then $u'-u\in\radix{V}$ and $v'-v\in\radix{V}$. $\metric{u',v'}=\metric{u,v}+\metric{u'-u,v}+\metric{u,v'-v}+\metric{u'-u,v'-v}=\metric{u,v}$.
\end{proof}

\begin{proposition}
\label{thm_cometric_signature}
The inner product $\annihg$ from above is non-degenerate, and if $g$ has the signature $(r,s,t)$, then the signature of $\annihg$ is $(0,s,t)$.
\end{proposition}
\begin{proof}
Let's take an orthonormal basis $(e_a)_{a=1}^n$, as in Definition \ref{def_inner_prod}. We have $\annih{e_a}=0$ for $a\in\{1,\ldots,r\}$, and the $1$-forms $\omega_a:=\annih{e_{r+a}}$ for $a\in\{1,\ldots,s+t\}$ are the generators of $\annih{V}$. They satisfy $\annihprod{\omega_a,\omega_b}=\metric{e_{r+a},e_{r+b}}$. Therefore, $(\omega_a)_{a=1}^{s+t}$ are linear independent and the signature of $\annihg$ is $(0,s,t)$.
\end{proof}

\subsection{The factor inner product space}
\label{s_factor_inner_prod}

In this section are applied general-known notions recalled in \sref{s_dual_subspace} to the case of degenerate inner product spaces.

\begin{remark}\rm
As it is known in the literature (\cf \eg \citep{Rom08}{87}, \citep{hal99}{33--35}) given a vector subspace $W$ of a vector space $V$, there is a canonical quotient, or factor vector space $V/W$. This applies as well to the degenerate part of $(V,g)$, which can be factored out, all of its other properties being preserved \seepeg{Kup87b}{274}.
\end{remark}

\begin{definition}\seepeg{Kup87b}{274}
\label{def_factrad}
Since $\radix{V}\subseteq V$ is a vector subspace, we can define a \textit{factor vector space} $\coannih{V}:=V/\radix{V}$. On $\coannih V$ we can define the \textit{factor inner product} $\coannihg$, by
\begin{equation}
	\coannihg(\coannih u, \coannih v)=\coannihprod{\coannih u, \coannih v}:=\metric{u,v},
\end{equation}
where $\coannih u, \coannih v$ are the equivalence classes of $u,v\in V$. The inner product $\coannihg$ is well defined, because $g|_{\radix{V}}\equiv 0$. It is also non-degenerate, because it is obtained by factoring out its degenerate part. The obtained factor inner product space $(\coannih{V},\coannihg)$ is in a tighter relation with $(\annih{V},\annihg)$ than the original space $(V,g)$.
\end{definition}

The following is a direct application of the Lemma \ref{thm_dual_subspace}.
\begin{proposition}
\label{thm_duality_annih_coannih}
$\annih V^*=\coannih V$ and $\coannih V^*=\annih V$.
\end{proposition}
\begin{proof}
By Definition \ref{def_factrad} $\coannih{V}:=V/\radix{V}$, and by Definition \ref{def_radical_annihilator}, $\annih{V}:=\subannih{\radix V}$. We apply the Lemma \ref{thm_dual_subspace} and obtain the desired duality.
\end{proof}

\begin{remark}\rm
Any element $\omega\in\annih{V}\subseteq V$ can be viewed as linear form on both $V$ and $\coannih V$, because for any $v,v'\in\coannih v\in \coannih V$, $\omega(v)=\omega(v')=\omega(\coannih v)$. Any element $v$ of any $\coannih v\in \coannih V$ can be viewed as a linear form on $\annih V$ by $\coannih v(\omega):=\omega(\coannih v)$, for any $\omega\in\annih V$.
\end{remark}

\begin{proposition}
\label{thm_iso_factrad_annih}
There is a canonical isomorphism $\factflat:\coannih V\stackrel\cong\to \annih{V}$, defined by
\begin{equation}
	\factflat(\coannih v):=\annih v,
\end{equation}
where $\coannih v:=v+\radix{V}$ is the equivalence class of a vector $v\in V$ defined by the subspace $\radix{V}\subseteq V$.
\end{proposition}
\begin{proof}
Let $v,v'\in V$. Then, $\coannih{v}=\coannih{v'}$ iff $v'-v\in \radix{V}$ iff $\annih{v}=\annih{v'}$, so $\factflat(\coannih v):=\annih{v}$ is well defined.
\end{proof}

\begin{definition}[see \eg \citep{Gibb06}{15}; \citep{GHLF04}{72}]
\label{def_annih_index_raise}
We define $\factsharp:\annih V\stackrel\cong\to \coannih{V}$ as $\factsharp:=\factflat^{-1}$, the inverse of the isomorphism $\factflat:\coannih V\stackrel\cong\to \annih{V}$ from Proposition \ref{thm_iso_factrad_annih}. For any $\annih v\in\annih V$, we have $\factsharp(\annih v)=\coannih v$. We call $\factflat:\coannih V\stackrel\cong\to \annih{V}$ the \textit{(index) lowering isomorphism} on $\coannih V$, and $\factsharp:\annih V\stackrel\cong\to \coannih{V}$ the \textit{(index) raising isomorphism} on $\coannih V$.
\end{definition}

\begin{proposition}
\label{thm_annih_index_raise}
The isomorphism $\factsharp:\annih V\stackrel\cong\to \coannih{V}$ takes the explicit form
\begin{equation}
	\factsharp(\tau)\omega:=\annihprod{\tau,\omega},
\end{equation}
for any $\tau,\omega\in\annih{V}$.
\end{proposition}
\begin{proof}
Let $\tau,\omega\in\annih{V}$ and $u,v\in V$ so that $\tau=\annih{u}$ and  $\omega=\annih{v}$. Then $\tau=\factflat(\coannih u)$ and  $\omega=\factflat(\coannih v)$. We have $\factsharp(\tau)\omega=\coannih u(\annih v)=\annih v(\coannih u)=\annih v(u)=\metric{u,v}=\annihprod{\tau,\omega}$.
\end{proof}

\begin{remark}\rm
Unfortunately, the index raising isomorphism $\factsharp$ doesn't lead to a vector in $V$, but in $\coannih{V}$, so it is not really a genuine index raising operator. We can always modify it to return a vector in $V$, but this construction is not unique (see \sref{s_factor_inner_prod_ext}).
\end{remark}

\begin{remark}\rm
The following equalities hold:
\begin{equation}
\begin{array}{lllllllllll}
\factsharp(\annih{u})(\annih{v}) &=& \coannih{u}(\annih{v}) &=& \annih{v}(\coannih{u}) &=& \annihprod{\annih{v},\annih{u}} &=& \annihprod{\annih{u},\annih{v}} \\
&=& \annih{u}(\coannih{v}) &=& \coannih{v}(\annih{u}) &=& \factflat(\coannih{v})(\annih{u}) &=& \factflat(\coannih{u})(\annih{v}) \\
&=& \factflat(\coannih{v})(u) &=& \factflat(\coannih{u})(v) &=& \flat(v)(u) &=& \flat(u)(v) \\
&=& \flat(v)(\annih{u}) &=& \flat(u)(\annih{v}) &=& \metric{u,v} &=& \metric{v,u} \\
&=& \coannihprod{\coannih{u},\coannih{v}} &=& \coannihprod{\coannih{v},\coannih{u}} &=& \factsharp(\annih{v})(\annih{u}).
\end{array}
\end{equation}
For convenience we may use the notations $\coannih u\annih v$ and $\annih u\coannih v$ for the same values. These identities and the Proposition \ref{thm_duality_annih_coannih} lead to the following:
\end{remark}

\begin{theorem}
$(\annih{V},\annihg)^*=(\coannih{V},\coannihg)$
\end{theorem}


\begin{remark}\rm
Because $\coannih{V}:=\annih{V}^*=\frac V{\radix V}$, the following sequence is exact:
\begin{center}
\begin{tikzpicture}
\matrix (m) [matrix of math nodes, row sep=4em,
column sep=4em, text height=1.5ex, text depth=0.25ex]
{ 0  & \radix{V} & V & \coannih{V} & 0 \\};
\path[right hook->]
(m-1-1) edge (m-1-2)
(m-1-2) edge node[auto] {$\radix{i}$}(m-1-3);
\path[->>]
(m-1-3) edge node[auto] {$\coannih{\pi}$}(m-1-4)
(m-1-4) edge(m-1-5);
\end{tikzpicture}
\end{center}
\end{remark}

\begin{remark}\rm
There is a similar construction for the inclusion $\annih{V}\subseteq V^*$, leading to the factor vector space
$$\coradix{V}:=\frac{V^*}{\annih{V}}.$$
We can apply Lemma \ref{thm_dual_subspace} to obtain a diagram like that in the Remark \ref{rem_dual_subspace}:
\begin{center}
\begin{tikzpicture}
\matrix (m) [matrix of math nodes, row sep=4em,
column sep=4em, text height=1.5ex, text depth=0.25ex]
{ 0  & \radix{V} & V & \coannih{V} & 0 \\
0	& \coradix{V} & V^*	& \annih{V} & 0	\\ };
\path[right hook->]
(m-1-1) edge (m-1-2)
(m-1-2) edge node[auto] {$\radix{i}$}(m-1-3);
\path[->>]
(m-1-3) edge node[auto] {$\coannih{\pi}$}(m-1-4)
(m-1-4) edge(m-1-5);
\path[->>]
(m-2-2) edge (m-2-1)
(m-2-3) edge node[auto,swap] {$\coradix{\pi}$}(m-2-2);
\path[left hook->]
(m-2-4) edge node[auto,swap] {$\annih{i}$}(m-2-3)
(m-2-5) edge(m-2-4);
\path[dashed]
(m-1-2) edge node[auto] {$\ast$} (m-2-2)
(m-1-3) edge node[auto] {$\ast$} (m-2-3)
(m-1-4) edge node[auto] {$\ast$} (m-2-4);
\end{tikzpicture}
\end{center}
where $\coannih{V}=\annih{V}^*=\frac V{\radix V}$ and $\coradix{V}=\radix{V}^*=\frac {V^*}{\annih V}$. If we include the isomorphisms due to the inner product, we obtain the diagram:
\begin{center}
\begin{tikzpicture}
\matrix (m) [matrix of math nodes, row sep=4em,
column sep=4em, text height=1.5ex, text depth=0.25ex]
{ 0  & \radix{V} & (V,g) & (\coannih{V},\coannihg) & 0 \\
0	& \coradix{V} & V^*	& (\annih{V},\annihg) & 0	\\ };
\path[right hook->]
(m-1-1) edge (m-1-2)
(m-1-2) edge node[auto] {$\radix{i}$}(m-1-3);
\path[->>]
(m-1-3) edge node[auto] {$\coannih{\pi}$}(m-1-4)
(m-1-4) edge(m-1-5);
\path[->>]
(m-2-2) edge (m-2-1)
(m-2-3) edge node[auto,swap] {$\coradix{\pi}$}(m-2-2)
(m-1-3) edge node[auto,sloped] {$\flat_V$}(m-2-4);
\path[left hook->]
(m-2-4) edge node[auto,swap] {$\annih{i}$}(m-2-3)
(m-2-5) edge(m-2-4);
\path[->]
(m-1-4) edge[bend right=15] node[auto,left] {$\flat$} (m-2-4)
(m-2-4) edge[bend right=15] node[auto,right]{$\sharp$} (m-1-4);
\end{tikzpicture}
\end{center}
\end{remark}

\subsection{Extensions of the dual inner product}
\label{s_factor_inner_prod_ext}

In a finite dimensional vector space $V$, for any subspace $W\subseteq V$ there is another subspace $W'$ so that $V=W\oplus W'$, which is not unique if $W\neq 0$ and $W\neq V$. We can apply this decomposition to a subspace and its annihilator. As a reference about such decompositions, please see \citep{hal99}{28--32} and \citep{Rom08}{103--104}. In particular, we can apply the known results about direct sums and annihilators to the radical and {\radannih} spaces associated to an inner product space.

\begin{proposition}
\label{thm_cometric_extension}
Let $\coradixinc{V}\subseteq V^*$ be a subspace such that $V^*=\annih{V}\oplus \coradixinc{V}$.
We can extend uniquely $\annihg$ to an inner product $g^*_{\coradixinc{V}}$ on the entire $V^*$ by the condition $\coradixinc{V}=\radix{(V^*,g^*_{\coradixinc{V}})}$.
\end{proposition}
\begin{proof}
Let $\omega,\tau\in V^*$, such that $\omega=\annih{\omega}+\coradixinc{\omega}$, $\tau=\annih{\tau}+\coradixinc{\tau}$, where $\annih{\omega},\annih{\tau}\in\annih{V}$ and $\coradixinc{\omega},\coradixinc{\tau}\in\coradixinc{V}$. We simply define $g^*_{\coradixinc{V}}$ as $g^*_{\coradixinc{V}}(\omega,\tau)=\annihprod{\annih{\omega},\annih{\tau}}$.
\end{proof}

\begin{remark}\rm
In the above proof, the definition of $g^*_{\coradixinc{V}}$ as $g^*_{\coradixinc{V}}(\omega,\tau)=\annihprod{\annih{\omega},\annih{\tau}}$ should not make us think that $g^*_{\coradixinc{V}}$ is independent on the choice of $\coradixinc{V}$. In reality, the decompositions $\omega=\annih{\omega}+\coradixinc{\omega}$ and $\tau=\annih{\tau}+\coradixinc{\tau}$ depend on $\coradixinc{V}$. There is a $1:1$ correspondence between such extensions of $\annihg$ and the choices of $\coradixinc{V}$.
\end{remark}

\begin{remark}\rm
Let $\flat^*_{\coradixinc{V}}:V^*\to V^{**}=V$ be the morphism induced by $g^*_{\coradixinc{V}}$, defined by $\flat^*_{\coradixinc{V}}(\omega)(\tau)=g^*_{\coradixinc{V}}(\omega,\tau)$ for all $\omega,\tau\in V^*$. 
Then	$\coradixinc{V}=\ker\flat^*_{\coradixinc{V}}$.
\end{remark}

\begin{proposition}
$V= \IM\flat^*_{\coradixinc{V}} \oplus \radix{V}$
\end{proposition}
\begin{proof}
Since $V^*=\annih{V} \oplus \ker\flat^*_{\coradixinc{V}}$, it follows that $\flat^*_{\coradixinc{V}}|_{\annih{V}}\to\IM\flat^*_{\coradixinc{V}}$ is an isomorphism. From Remark \ref{thm_img_ker_radix} we have that $V= \IM\flat^*_{\coradixinc{V}} \oplus \radix{V}$.
\end{proof}

\begin{proposition}
$\flat\circ\flat^*_{\coradixinc{V}}|_{\coradixinc{V}}=1_{\coradixinc{V}}.$
\end{proposition}
\begin{proof}
Let $(\omega_a)_{a=1}^{s+t}$ be a basis of $\annih{V}$. Then, $\varepsilon_a:=\flat^*_{\coradixinc{V}}(\omega_a)\in V^{**}= V$.

$\omega_b(\varepsilon_a)=\varepsilon_a(\omega_b)=\flat^*_{\coradixinc{V}}(\omega_a)(\omega_b)=g^*_{\coradixinc{V}}(\omega_a,\omega_b)$. But $g^*_{\coradixinc{V}}(\omega_a,\omega_b)=\annihg(\omega_a,\omega_b)$, because $g^*_{\coradixinc{V}}|_{\annih{V}}=\annihg$, so $\omega_b(\varepsilon_a)=\annihg(\omega_a,\omega_b)$. Let $(e_a)_{a=1}^{s+t}$ chosen such that $\flat(e_a)=\omega_a$. They satisfy $g(e_a,e_b)=\annihg(\omega_a,\omega_b)$. We have $\flat(e_a)(e_b)=\omega_a(e_b)=\annihg(\omega_a,\omega_b)$, so $\omega_b(\varepsilon_a-e_a)=0$ for all $a,b\in\{1,\ldots,s+t\}$. It follows that $\varepsilon_a-e_a\in\radix{V}$, and therefore $\flat(\varepsilon_a)=\flat(e_a)=\omega_a$, for all $a\in\{1,\ldots,s+t\}$. Therefore, $\flat\circ\flat^*_{\coradixinc{V}}|_{\coradixinc{V}}=1_{\coradixinc{V}}$.
\end{proof}

\begin{remark}\rm
\label{thm_metric_extension}
Instead of the construction in Proposition \ref{thm_cometric_extension}, we can start by choosing a subspace $\coannihinc{V}\subseteq V$ such that $V=\coannihinc{V}\oplus\radix{V}$. It follows that $g|_{\coannihinc{V}}$ is non-degenerate, and we have an isomorphism  $\flat|_{\coannihinc{V}}:\coannihinc{V}\stackrel\cong\to\annih{V}$. We can identify thus $\coannihinc{V}$ with $\coannih{V}$ and with $\annih{V}^*$, and  $\annihg$ with the dual of $g|_{\coannihinc{V}}$. We can consider this way that $\annihg\in \coannihinc{V}\odot \coannihinc{V}\subseteq V\odot V$, for a given choice of $\coannihinc{V}$ or, equivalently, of $\coradixinc{V}$. The relation between the two choices is given by $\coannihinc{V}=\IM\flat^*_{\coradixinc{V}}$.
In other words, it is enough to know the inclusion morphism $\annih{V}^*\hookrightarrow V$.
\end{remark}

\begin{definition}\cfeg{Pam03}{268}
\label{def_screen_space}
A vector subspace $\coannihinc{V}\subseteq V$ like in Remark \ref{thm_metric_extension} is named a \textit{screen space} for the space $\annih{V}$, being a realization of its dual.
\end{definition}

\begin{remark}\rm
These constructions can be used to raise indices, but they are not unique, depending on the choice of $\coannihinc{V}$ or $\coradixinc{V}$. The operation of raising index defined with their help is not an invariant operation of the degenerate inner product space $(V,g)$. On the other hand, the spaces $\coannih{V}$ and $\coradix{V}$ are invariants, and satisfy $\coannih{V}\cong\coannihinc{V}$ and $\coradix{V}\cong\coradixinc{V}$.
\end{remark}

\subsection{Radical and radical-annihilator bases}
\label{s_radical_bases}

For explicit calculations involving the radical annihilator $\annih{V}$ of an inner product space $(V,g)$ or its {\radannih} inner product $\annihg$, it is useful to have a basis of $\annih{V}$. To a basis of $\annih{V}$ corresponds a unique dual basis of $\coannih{V}:=\annih{V}^*$ \cfeg{Rom08}{96}. Since there are many occasions when we perform simultaneously calculations on $V$ and $V^*$, it is useful to have a basis on $V^*$ which extends the basis on $\annih{V}$. We will see that such a basis turns out to be the dual of a basis on $V$ extending a basis on $\radix{V}$.

\begin{definition}
\label{def_radical_and_radical_annihilator_bases}
A \textit{radical basis} of $(V,g)$ is a basis obtained by extending a basis of $\radix{V}$ to the entire $V$ \cfeg{Pam03}{263, 268}. If the elements of the basis which do not belong to $\radix{V}$ are orthogonal (orthonormal), the basis is named \textit{orthogonal (orthonormal) radical basis}. A \textit{{\radannih} basis} of $(V^*,\annihg)$ is obtained by extending a basis of $\annih{V}$ to the entire $V^*$. If the elements of the basis from $\annih{V}$ are orthogonal (orthonormal), the basis is named \textit{orthogonal (orthonormal) {\radannih} basis}.
\end{definition}

\begin{remark}\rm
The components of $g$ in a basis $(e_a)_{a=1}^n$ of $V$ are given by $g_{ab}=\metric{e_a,e_b}$. The components of the {\radannih} inner product $\annihg$ are given, in a basis $(\omega^a)_{a=1}^{\rank g}$ of $\annih{V}$, by $\annihg^{ab}=\annihprod{\omega^a,\omega^b}$. We cannot regard the coefficients $\annihg^{ab}$ as being the components of a bilinear form on $V$, because $\annihg$ is in fact a bilinear form in $\coannih{V}\odot\coannih{V}$, and there is no canonical injection of $\annih{V}^*=\coannih{V}$ in $V$. Such a canonical injection does not exist, despite the fact that $\annih{V}\subseteq V^*$. It exists only in the special case when $g$ is non-degenerate. But, as we have seen in \sref{s_factor_inner_prod_ext}, we can extend $\annihg$ to an inner product on $V^*$ in a non-unique fashion.
\end{remark}

\begin{remark}\rm
If we have a {\radannih} basis of $V^*$, the elements of the basis induce a unique basis (the dual basis) on $V$, and the elements of the basis spanning $\annih{V}$ induce a unique basis on $\coannih{V}$ \cfeg{Rom08}{96}. This shows a relation between these bases.
\end{remark}

\begin{theorem}[Of the radical and {\radannih} dual bases]
\label{thm_radical_and_radical_annihilator_bases}
The dual of a radical basis is a {\radannih} basis, and conversely. The dual of an orthogonal (orthonormal) radical basis is an orthogonal (orthonormal) {\radannih} basis, and conversely.
\end{theorem}
\begin{proof}
The first part of the theorem follows from the Lemma \ref{thm_annihilator_dual_bases}, since $\annih{V}=\subannih{(\radix{V})}$.
 We prove the second part.
 
``$\Rightarrow$''.
Let's assume that $(e_a)_{a=1}^n$ is an orthogonal (orthonormal) radical basis. In this basis the matrix $(g_{ab})_{1\leq a,b\leq n}$ is diagonal: $g_{ab}=\alpha_a\delta_{ab}$. For $r+1\leq a \leq n$, $\tilde{\omega}^a:=\flat(e_a)$ form an orthogonal (orthonormal) basis for $\annih{V}$, and $\tilde{\omega}^a(e_b)=g(e_a,e_b)=\alpha_a \delta_{ab}$. Therefore, if $(\omega^a)_{a=1}^n$ is the dual basis of $(e_a)_{a=1}^n$, for $r+1\leq a \leq n$,
$$\omega^a=\frac 1 {\alpha_a}\tilde{\omega}^a.$$
Since
$$\annihprod{\omega^a,\omega^b}=\frac 1 {\alpha_a\alpha_b}\annihprod{\tilde{\omega}^a,\tilde{\omega}^b}=\frac 1 {\alpha_a\alpha_b}\metric{e_a,e_b}=\frac 1 {\alpha_a\alpha_b}\alpha_a\delta_{ab}=\frac 1 {\alpha_a}\delta_{ab},$$
the basis $(\omega^a)_{a=1}^n$ is {\radannih} and orthogonal (orthonormal).

``$\Leftarrow$''.
If $(\omega^b)_{b=1}^n$ is an orthogonal (orthonormal) {\radannih} basis, in this basis $\annihg$ has the form $\annihg^{ab}=\beta^a\delta^{ab}$. As in \sref{s_factor_inner_prod_ext}, by choosing $\coradixinc{V}\subseteq V^*$ so that $V^*=\annih{V}\oplus\coradixinc{V}$, we can construct the vectors $\varepsilon_a:=\flat^*_{\coradixinc{V}}(\omega^a)$. They satisfy $g(\varepsilon_a,\varepsilon_b)=\omega^b(\varepsilon_a)=\annihg(\omega^a,\omega^b)=\beta^a\delta^{ab}$, so they are orthogonal (orthonormal).
Let's construct the vectors
$$\tilde e_a:=\frac 1 {\beta^a}\varepsilon_a,$$
for $r<a\leq n$. They satisfy
$$g(\tilde e_a,\tilde e_b)=\frac 1 {\beta^a\beta^b}g(\varepsilon_a,\varepsilon_b)=\frac 1 {\beta^a\beta^b}\beta^a\delta_{ab}=\frac 1 {\beta^a} \delta_{ab}.$$
$\omega^b(\tilde e_a)=\omega^b(\frac 1 {\beta^a}\varepsilon_a)=\frac 1 {\beta^a}\omega^b(\varepsilon_a)=\frac 1 {\beta^a}\beta^a\delta_{ab}=\delta_{ab}$. But $\omega^b(e_a)=\delta_{ab}$ too, so $\omega^b(e_a-\tilde e_a)=0$ for all $r<a,b\leq n$.
It follows that $e_a-\tilde e_a\in\radix{V}$ for all $r<a\leq n$. Hence, $g(e_a,e_b)=g(\tilde e_a,\tilde e_b)=\frac 1{\beta^a\beta^b}g(\varepsilon_a,\varepsilon_b)=\frac 1{\beta^a\beta^b}\beta^a\delta_{ab}=\frac 1{\beta^a}\delta_{ab}$ for all $r<a,b\leq n$. Therefore, if $(\omega^b)_{b=1}^n$ is an orthogonal (orthonormal) basis, so is $(e_a)_{a=1}^n$.
\end{proof}

\begin{remark}\rm
If $\radix{V}$ is an invariant subspace for an operator $A\in\Aut(V)$, that is $A(\radix{V})=\radix{V}$, then, if $(e_a)_{a=1}^n$ is a radical basis of $V$, $(A(e_i))_{a=1}^n$ is a radical basis. The dual $A^*$ of $A$ is a vector space automorphism of $V^*$, and $A^*(\annih{V})=\annih{V}$. $A^*$ transforms a {\radannih} basis in another {\radannih} basis. If $A$ preserves the inner product $g$, then $A^*$ preserves $\annihg$. In this case, $A$ transforms a radical orthogonal (orthonormal) basis into a radical orthogonal (orthonormal) basis, and $A^*$  transforms a {\radannih} orthogonal (orthonormal) basis into a {\radannih} orthogonal (orthonormal) basis.
\end{remark}

\begin{remark}\rm
\label{thm_cometric_extension_basis}
Any radical basis (and, by Theorem \ref{thm_radical_and_radical_annihilator_bases}, any {\radannih} basis) can be used to extend the inner product $\annihg$ on $\annih{V}\subseteq V^*$ to an inner product $g^*$ defined on the entire $V^*$. We just take as $\coannihinc{V}\subseteq V$ the subspace generated by the vectors of the basis $(e_a)_{a=1}^n$ which are not totally degenerate, that is, $e_a\notin \radix{V}$, and as $\coradixinc{V}\subseteq V^*$ the subspace of $V^*$ generated by the covectors $\omega_a\notin \annih{V}$. Although the extension $g^*$ is not unique, it is uniquely defined given the basis. In practice, we can use $g^*$ instead of $\annihg$ even if it is not unique, as long as both its slots are contracted with elements or slots from $\annih{V}$. In a {\radannih} basis, the coefficients $g^{*ab}$ are the same for $a,b>r$, and coincide with $\annihg^{ab}$.
\end{remark}

\subsection{The radical-annihilator inner product in a basis}

The content of this section is a straight application of elementary linear algebra facts.

Let us consider an inner product space $(V,g)$, and an orthogonal radical basis $(e_a)_{a=1}^n$ of $V$ in which $g$ takes the diagonal form $g=\diag(\alpha_1,\alpha_2,\ldots,\alpha_n)$, $\alpha_a\in\reals$ for all $1\leq a\leq n$. By counting the number of coefficients $\alpha_a$ which are equal, greater or less than zero, we find the signature of $g$. The inner product satisfies:
\begin{equation}
g_{ab}=\metric{e_a,e_b}=\alpha_a\delta_{ab}.
\end{equation}
We also have
\begin{equation*}
\annih{e_a}(e_b):=\metric{e_a,e_b}=\alpha_a\delta_{ab},
\end{equation*}
and
\begin{equation}
\annih{e_a}=(\alpha_1\delta_{a1}\ \ldots\ \alpha_n\delta_{an})=\alpha_a(0\ \ldots\ 1\ \ldots \ 0)=\alpha_a (e_a)^T.
\end{equation}

\begin{proposition}
\label{thm_cometric_in_basis}
If in a basis the inner product has the form $g_{ab}=\alpha_a\delta_{ab}$, then 
\begin{equation}
\annihg^{ab}=\frac 1{\alpha_a}\delta^{ab},
\end{equation}
where $\alpha_a\neq 0$.
\end{proposition}
\begin{proof}
Since
\begin{equation*}
\annihprod{\annih{e_a},\annih{e_b}}=\metric{e_a,e_b}=\alpha_a\delta_{ab},
\end{equation*}
and in the same time
\begin{equation*}
\annihprod{\annih{e_a},\annih{e_b}}=\alpha_a \alpha_b \annihprod{(e_a)^T,(e_b)^T}=\alpha_a \alpha_b\annihg^{ab},
\end{equation*}
we have that
\begin{equation*}
\alpha_a \alpha_b\annihg^{ab}=\alpha_a \delta_{ab},
\end{equation*}
This leads, for $\alpha_a\neq 0$, to
\begin{equation*}
\annihg^{ab}=\frac 1{\alpha_a}\delta_{ab}.
\end{equation*}
The case when $\alpha_a = 0$ is not allowed, since $\annihg$ is defined only on $\IM\flat$.
\end{proof}

\begin{remark}\rm
We can extend $\annihg$ to the entire $V^*$, as in the Remark \ref{thm_cometric_extension_basis}. The extended inner product $g^*$ has the components
\begin{equation}
\begin{array}{llll}
g^*{}^{ab}&=&\frac 1{\alpha_a}\delta^{ab}&\tn{ for }\alpha_a\neq 0,\tn{ and} \\
g^*{}^{ab}&=&0&\tn{ for }\alpha_a = 0.
\end{array}
\end{equation}
\end{remark}

\section{Tensors on degenerate inner product spaces}
\label{s_tensors_inner_prod}

In the following, we will be interested in some of the elementary properties of tensors obtained from the invariant spaces associated to an inner product space $(V,g)$: mainly $V$, $V^*$, $\radix{V}$ and $\annih{V}$, but we will need $\coradix{V}$ and $\coannih{V}$ too. When constructing the various tensor products and study their properties, we need to remember the relations $\annih{V}^*\cong\coannih{V}$ and $\coradix{V}=\radix{V}^*$, as well as the inclusions $\radix{V}\subseteq V$ and $\annih{V}\subseteq V^*$.

The main class of tensor spaces associated to $(V,g)$ is given by:
\begin{definition}\cfeg{ONe83}{35}
A \textit{tensor} of type $(r,s)$ is an element of the vector space 
\begin{equation}
	\tensors r s V:= V^{\otimes r}\otimes V^*{}^{\otimes s}.
\end{equation}
\end{definition}

For such a tensor we can define contractions between an upper and a lower index:
\begin{definition}\cfeg{ONe83}{40}
Let $T\in\tensors r s V$, $1\leq k \leq r$, $1\leq l \leq s$.
We denote by $C^k_l(T)$ the \textit{contraction} between the $k$-th contravariant slot and the $l$-th covariant slot of $T$,   $C^k_l(T)\in\tensors {r-1} {s-1} V$,
\begin{equation}
	(C^k_l(T))^{a_1\ldots \widehat{a_k}\ldots a_r}{}_{b_1\ldots\widehat{b_l}\ldots b_s}:=\sum_{i=0}^n T^{a_1\ldots i\ldots a_r}{}_{b_1\ldots i\ldots b_s},
\end{equation}
which is independent on the chosen basis.
\end{definition}

\subsection{Radical and radical-annihilator tensors}
\label{s_radical_annih_tensors}

The properties of the subspaces $\radix{V}\subseteq V$ and $\annih{V}\subseteq V^*$ suggest that the tensors which have arguments restrained to these subspaces also have distinguishing properties.

\begin{definition}
\label{def_radix_annih_tensor}
Let $T$ be a tensor of type $(r,s)$. We call it \textit{radical} in the $k$-th contravariant slot if $T\in V^{\otimes {k-1}}\otimes\radix{V}\otimes V^{\otimes {r-k}}\otimes {V^*}^{\otimes s}$. We call it \textit{{\radannih}} in the $k$-th covariant slot if  $T\in V^{\otimes r}\otimes {V^*}^{\otimes {k-1}}\otimes\annih{V}\otimes {V^*}^{\otimes s-k}$.
\end{definition}

\begin{proposition}
\label{thm_radical_contravariant_index}
A tensor $T\in\tensors r s V$ is radical in the $k$-th contravariant slot if and only if its contraction $C^k_{s+1}(T\otimes\omega)$ with any {\radannih} linear $1$-form $\omega\in V^*$ is zero.
\end{proposition}
\begin{proof}
For simplicity, we can consider $k=r$ (if $k<r$, we can make use of the permutation automorphisms of the tensor space $\tensors r s V$). 
T can be written as a sum of linear independent terms having the form $\sum_{\alpha}S_{\alpha}\otimes v_{\alpha}$, with $S_{\alpha}\in\tensors{r-1}s V$ and $v_{\alpha}\in V$. We keep only the terms with $S_{\alpha}\neq 0$. The contraction of the $r$-th contravariant slot with any $\omega\in\annih{V}$ becomes $\sum_{\alpha}S_{\alpha}\omega(v_{\alpha})$.

If $T$ is radical in the $r$-th contravariant slot, for all $\alpha$ and any $\omega\in\annih{V}$ we have $\omega(v_{\alpha})=0$, therefore $\sum_{\alpha}S_{\alpha}\omega(v_{\alpha})=0$.

Reciprocally, if $\sum_{\alpha}S_{\alpha}\omega(v_{\alpha})=0$, it follows that for any $\alpha$, $S_{\alpha}\omega(v_{\alpha})=0$. Then, $\omega(v_{\alpha})=0$, because $S_{\alpha}=0$. It follows that $v_{\alpha}\in\radix{V}$.
\end{proof}

\begin{proposition}
\label{thm_radical_annihilator_covariant_index}
A tensor $T\in\tensors r s V$ is {\radannih} in the $k$-th covariant slot if and only if its $k$-th contraction with any totally degenerate vector is zero.
\end{proposition}
\begin{proof}
The proof goes as in Proposition \ref{thm_radical_contravariant_index}.
\end{proof}

\begin{example}
\label{thm_metric_radical_annihilator}
The inner product $g$ is {\radannih} in both its slots. This means that $g\in\annih{V}\odot\annih{V}$.
\end{example}
\begin{proof}
Follows directly from the definition of $\radix{V}$ and of {\radannih} tensors.
\end{proof}

\begin{proposition}
\label{thm_radical_annihilator_vs_radical_contraction}
The contraction between a radical slot and a {\radannih} slot of a tensor is zero.
\end{proposition}
\begin{proof}
Follows from the Proposition \ref{thm_radical_contravariant_index} combined with the commutativity between tensor products and linear combinations with contraction. The proof goes similar to that of the Proposition \ref{thm_radical_contravariant_index}.
\end{proof}

\subsection{Index lowering}
\label{s_tensors_index_lower}

The inner product $g$ allows us to lower indices, in a similar manner to the non-degenerate case \seepeg{ONe83}{60}. More precisely:

\begin{definition}
\label{def_index_lower}
If
\begin{equation}
	T\in \tensors r s V:=V^{\otimes r}\otimes {V^*}^{\otimes s}
\end{equation}
is a tensor over $V$, with $r\geq 1,s\geq 0$, then the inner product $g\in\tensors 0 2 V$ defines, for $k\in\{1,\ldots,r\}$, a new tensor $\idxlow_k(T)\in \tensors{r-1}{s+1} V$ by contraction:
\begin{equation}
	\idxlow_k(T) :=C^k_{s+2} (T\otimes g),
\end{equation}
which in a frame takes the form:
\begin{equation}
	\idxlow_k(T)^{a_1\ldots\widehat{a}_k\ldots a_r}{}_{b_1\ldots b_s b_{s+1}} := T^{a_1\ldots a_k\ldots a_r}{}_{b_1\ldots b_s} g_{b_{s+1}a_k}.
\end{equation}
\end{definition}

\begin{remark}\rm
During this process, some information is lost, if $g$ is degenerate, as the following proposition shows.
\end{remark}

\begin{proposition}
\label{thm_index_lowering}
Let $T\in \tensors r s V:=V^{\otimes r}\otimes {V^*}^{\otimes s}$ be a tensor over $V$. Then $\idxlow_k(T)=0$ if and only if $T\in V^{\otimes {k-1}}\otimes\radix{V}\otimes V^{\otimes {r-k}}\otimes {V^*}^{\otimes s}$.
\end{proposition}
\begin{proof}
According to the Example \ref{thm_metric_radical_annihilator}, the inner product is {\radannih}. Hence, we obtain the desired result as a consequence of the Proposition \ref{thm_radical_annihilator_covariant_index}.
\end{proof}

\begin{corollary}
A tensor $T\in \tensors r s V:=V^{\otimes r}\otimes {V^*}^{\otimes s}$ over $V$ can be recovered from $\idxlow_k(T)$ only up to a tensor $T'\in V^{\otimes {k-1}}\otimes\radix{V}\otimes V^{\otimes {r-k}}\otimes {V^*}^{\otimes s}$. In other words,
\begin{equation}
	\ker(\idxlow_k:\tensors r s V\to\tensors{r-1}{s+1} V)=V^{\otimes {k-1}}\otimes\radix{V}\otimes V^{\otimes {r-k}}\otimes {V^*}^{\otimes s}.
\end{equation}
\end{corollary}
\begin{proof}
Follows immediately from Proposition \ref{thm_index_lowering}.
\end{proof}

\begin{remark}\rm
\label{rem_index_raising}
The $\idxlow_k$ operator loses information, because it is not invertible. Consequently, an index raising operator $\idxraise^k$ cannot be properly defined in a unique way. But we can define a non-canonical index raising operator $\idxraise^k$ if we use a screen space (Definition \ref{def_screen_space}, see also \citep{Kup87b}{262--263} and \citep{Pam03}{268}).
\end{remark}

\subsection{Covariant contraction}
\label{s_tensors_covariant_contraction}

We don't need an inner product to define contractions between one covariant and one contravariant indices. We can use the inner product $g$ to contract between two contravariant indices, obtaining the \textit{contravariant contraction operator} $C^{kl}$ \cfeg{ONe83}{83}. On the other hand, the contraction is not always well defined for two covariant indices. We will see that we can use $\annihg$ for such contractions, but this works only for {\radannih} covariant vectors or covariant slots. Fortunately, this kind of tensors turn out to be the relevant ones in the applications to singular {\semiriem} geometry.

\begin{definition}
\label{def_trace_covariant}
We can define uniquely the \textit{covariant contraction} or \textit{covariant trace} operator by the following steps.
\begin{enumerate}
	\item 
We define it first on tensors $T\in\annih{V}\otimes\annih{V}$, by $C_{12}T=\annihg^{ab}T_{ab}$. This definition is independent on the basis, because $\annihg\in\coannih{V}\otimes\coannih{V}$. 
	\item 
Let $T\in\tensors r s V$ be a tensor with $r\geq 0$ and $s\geq 2$, which satisfies $T\in V^{\otimes r}\otimes {V^*}^{\otimes {s-2}}\otimes\annih{V}\otimes\annih{V}$, that is, $T(\omega_1,\ldots,\omega_r,v_1,\ldots,v_s)=0$ for any $\omega_i\in V^*, i=1,\ldots,r$, $v_j\in V,j=1,\ldots,s$ whenever $v_{s-1}\in\radix{V}$ or $v_{s}\in\radix{V}$. Then, we define the covariant contraction between the last two covariant slots by the operator
\begin{equation*}
C_{s-1\,s}:=1_{\tensors r {s-2} V}\otimes\annihg:\tensors r s V\otimes \annih{V}\otimes\annih{V}\to\tensors r {s-2} V,
\end{equation*}
where $1_{\tensors r {s-2} V}:\tensors r {s-2} V\to\tensors r {s-2}V$ is the identity.
In a radical basis, the contraction can be expressed by
\begin{equation*}
\label{eq_trace_covariant_end}
(C_{s-1\,s} T)^{a_1\ldots a_r}{}_{b_1\ldots b_{s-2}} :=	\annihg^{b_{s-1} b_{s}}T^{a_1\ldots a_r}{}_{b_1\ldots \ldots b_{s-2}b_{s-1}b_{s}}.
\end{equation*}
	\item 
Let $T\in\tensors r s V$ be a tensor with $r\geq 0$ and $s\geq 2$, which satisfies $T\in V^{\otimes r}\otimes {V^*}^{\otimes {k-1}}\otimes\annih{V}\otimes {V^*}^{\otimes l-k-1}\otimes\annih{V}\otimes {V^*}^{\otimes s-l}$, $1\leq k<l\leq s$, that is, $T(\omega_1,\ldots,\omega_r,v_1,\ldots,v_k,\ldots,v_l,\ldots,v_s)=0$ for any $\omega_i\in V^*, i=1,\ldots,r$, $v_j\in V,j=1,\ldots,s$ whenever $v_k\in\radix{V}$ or $v_l\in\radix{V}$. We define the contraction
\begin{equation*}
C_{kl}:V^{\otimes r}\otimes {V^*}^{\otimes {k-1}}\otimes\annih{V}\otimes {V^*}^{\otimes l-k-1}\otimes\annih{V}\otimes {V^*}^{\otimes s-l}
\to V^{\otimes r}\otimes {V^*}^{\otimes {s-2}},
\end{equation*}
by $C_{kl}:=C_{s-1\,s}\circ P_{k,s-1;l,s}$, where $C_{s-1\,s}$ is the contraction defined above, and $P_{k,s-1;l,s}:T\in\tensors r s V\to T\in\tensors r s V$ is the permutation isomorphisms which moves the $k$-th and $l$-th slots in the last two positions. In a basis, the components take the form
\begin{equation}
\label{eq_trace_covariant}
(C_{kl} T)^{a_1\ldots a_r}{}_{b_1\ldots\widehat{b}_k\ldots\widehat{b}_l\ldots b_s} :=	\annihg^{b_k b_l}T^{a_1\ldots a_r}{}_{b_1\ldots b_k\ldots b_l\ldots b_s}.
\end{equation}\end{enumerate}
We denote the contraction $C_{kl} T$ of $T$ also by 
\begin{equation*}
C(T(\omega_1,\ldots,\omega_r,v_1,\ldots,\cocontr,\ldots,\cocontr,\ldots,v_s))
\end{equation*}
or simply
\begin{equation*}
T(\omega_1,\ldots,\omega_r,v_1,\ldots,\cocontr,\ldots,\cocontr,\ldots,v_s).
\end{equation*}
\end{definition}

\begin{theorem}
\label{thm_contraction}
Let $T\in\tensors r s V$, $s\geq 2$, be a tensor which is {\radannih} in the $k$-th and $l$-th covariant slots ($1\leq k<l \leq n$). Let $(e_a)_{a=1}^n$ be a radical orthogonal basis, so that $e_1,\ldots,e_{r_g}\in\radix{V}$, where $r_g=n-\rank g$. Then
\begin{equation}
\begin{array}{l}
T(\omega_1,\ldots,\omega_r,v_1,\ldots,\cocontr,\ldots,\cocontr,\ldots,v_s) =\\
\sum_{a=r_g+1}^n \ds{\frac{1}{\metric{e_a,e_a}}}T(\omega_1,\ldots,\omega_r,v_1,\ldots,e_a,\ldots,e_a,\ldots,v_s),
\end{array}
\end{equation}
for any $v_1,\ldots,v_s,\omega_1,\ldots,\omega_r$.
\end{theorem}
\begin{proof}
The dual basis of $(e_a)_{a=1}^n$ is, according to Theorem \ref{thm_radical_and_radical_annihilator_bases}, an orthogonal {\radannih} basis. Therefore, $\annihg$ is diagonal. From the Proposition \ref{thm_cometric_in_basis} we recall that $\annihg^{aa}=\ds{\frac 1{g_{aa}}}$, for $a>r_g$. Therefore
\begin{equation*}
\begin{array}{l}
\annihg^{ab}T(\omega_1,\ldots,\omega_r,v_1,\ldots,e_a,\ldots,e_b,\ldots,v_s)=\\
\ds{\sum_{a=r_g+1}^n \frac{1}{\metric{e_a,e_a}}}T(\omega_1,\ldots,\omega_r,v_1,\ldots,e_a,\ldots,e_a,\ldots,v_s)$$
\end{array}
\end{equation*}
and we obtain the desired result.
\end{proof}

\begin{lemma}
\label{thm_contraction_with_metric_vector_spaces}
If $T$ is a tensor $T\in\tensors r s V$ with $r\geq 0$ and $s\geq 1$, which is {\radannih} in the $k$-th covariant slot, $1\leq k\leq s$, then its contraction with the inner product gives the same tensor:
\begin{equation}
\label{eq_contraction_with_metric}
T(\omega_1,\ldots,\omega_r,v_1,\ldots,\mathop{\cocontr}_{k},\ldots,v_s)\metric{v_k,\cocontr}=T(\omega_1,\ldots,\omega_r,v_1,\ldots,v_k,\ldots,v_s)
\end{equation}
\end{lemma}
\begin{proof}
Let's first consider the case when $T\in\tensors 0 1 V$, in fact, $T=\omega\in\annih V$. Then, equation \eqref{eq_contraction_with_metric} reduces to
\begin{equation}
	\omega(\cocontr)\metric{v,\cocontr}=\omega(v).
\end{equation}
But since $\omega\in\annih V$, it takes the form $\omega=\annih{u}$ for $u\in V$, and $\omega(\cocontr)\metric{v,\cocontr}=\annihprod{\omega,\annih v}=\metric{u,v}=\annih u (v)=\omega(v)$.

The general case is obtained from the linearity of the tensor product in the $k$-th covariant slot.
\end{proof}

\begin{corollary}
\label{thm_contracted_metric_w_metric}
$\metric{v,\cocontr}\metric{w,\cocontr}=\metric{v,w}.$
\end{corollary}
\begin{proof}
Follows from Lemma \ref{thm_contraction_with_metric_vector_spaces} and Example \ref{thm_metric_radical_annihilator}.
\end{proof}

\begin{example}
\label{thm_contracted_metric_w_itself}
$\metric{\cocontr,\cocontr}=\rank g.$
\end{example}
\begin{proof}
We recall that $g\in\annih{V}\odot\annih{V}$, $\annihg\in\annih{V}^*\odot\annih{V}^*$. When restricted to $\annih{V}$ and $\coannih{V}$ they are non-degenerate and inverse to one another. Since $\dim\annih{V}=\dim\ker\flat=\rank g$, we obtain $\metric{\cocontr,\cocontr}=\rank g$.
\end{proof}

\begin{remark}\rm
There are two ways to contract between the $k$-th contravariant slot and the $l$-th covariant slot of a tensor $T\in\tensors r s V$. The usual one, $C^k_l$, does not involve the inner product or its dual. The second way is obtained by first lowering the contravariant slot using $\flat_k$, then contracting it with the other covariant slot using $C_{l\,s+1}$. It is not always defined, but only when the $l$-th covariant slot of $C_{l\,s+1}(\flat_k(T))$ is {\radannih}. This happens in particular when the inner product $g$ is non-degenerate. Please note that it is possible the $l$-th covariant slot to become {\radannih} after the lowering of $k$-th contravariant slot, although before it was not, as we can see from the example $T=v\otimes \omega+w\otimes\tau$, where $v\in V-\radix{V}$, $w\in\radix{V}$, $\omega\in\annih{V}$, $\tau\in V^*-\annih{V}$. Lowering the contravariant index leads to $\flat_1(T)=\annih v\otimes \omega+\annih w\otimes \tau=\annih v\otimes \omega\in \annih{V}\otimes\annih{V}$, and we can contract with $C_{12}$ to obtain $\annihprod{\annih v,\omega}=\omega(v)$. Contracting directly by $C^1_1$ leads to $C^1_1(v\otimes \omega+w\otimes\tau)=\omega(v)+\tau(w)$, which is different, because $\tau(w)\neq 0$.
\end{remark}

\section{Conclusions and perspectives}

We have seen that we can extend operations which are usually associated to non-degenerate inner product, to the degenerate case. The central operation of this kind is the covariant contraction for special cases of tensors. The main opening provided by these extensions is explored in subsequent articles, where we applied them to construct various invariants in singular {\semiriem} geometry \cite{Sto11a,Sto11b,Sto11d}. We considered both the case when the signature of the inner product is constant, and when it is variable. Then, we applied it to the study of some singularities in the theory of General Relativity \cite{Sto11e,Sto11f,Sto11g,Sto11c,Sto11h}.



\end{document}